\documentclass[a4paper,UKenglish]{lipics-v2016}
 
\usepackage{microtype}

\usepackage[disable]{todonotes}

\title{Boundary Labeling for Rectangular Diagrams\footnote{Research of PB, MK and DM is supported in part by Natural Sciences and Engineering Research Council of Canada (NSERC). SM is supported by a Carleton-Fields postdoctoral fellowship.}}
\titlerunning{Boundary Labeling of Rectangular Diagrams} 

\author[1]{Prosenjit Bose}
\author[2]{Paz Carmi}
\author[3]{J. Mark Keil}
\author[4]{Saeed Mehrabi}
\author[5]{Debajyoti Mondal}
\affil[1]{School of Computer Science, Carleton University, Ottawa, Canada.\\
  \texttt{jit@scs.carleton.ca}}
\affil[2]{Department of Computer Science, Ben-Gurion University, Beer-Sheva, Israel.\\
  \texttt{carmip@cs.bgu.ac.il}}
\affil[3]{Department of Computer Science, University of Saskatchewan, Saskatoon, Canada.
  \texttt{mark.keil@usask.ca}}
\affil[4]{School of Computer Science, Carleton University, Ottawa, Canada.\\
  \texttt{saeed.mehrabi@carleton.ca}}
\affil[5]{Department of Computer Science, University of Saskatchewan, Saskatoon, Canada.
  \texttt{d.mondal@usask.ca}}

\authorrunning{Bose et al.} 

\Copyright{P. Bose, P. Carmi, J. M. Keil, S. Mehrabi, D. Mondal}


\EventEditors{}
\EventNoEds{2}
\EventLongTitle{}
\EventShortTitle{arXiv}
\EventAcronym{}
\EventYear{}
\EventDate{}
\EventLocation{}
\EventLogo{}
\SeriesVolume{}
\ArticleNo{No}

\newcommand\len[1]{|#1|} 

\begin{document}

\maketitle

\begin{abstract}
Given a set of $n$ points (sites) inside a rectangle $R$ and $n$ points (label locations or ports) on its boundary, a boundary labeling problem seeks  ways of connecting every site to a distinct port while achieving different labeling aesthetics. We examine the scenario when the connecting lines (leaders)  are drawn as axis-aligned polylines with few bends, every leader lies strictly inside $R$, no two leaders cross, and the sum of the lengths of all the leaders is minimized.  In a  $k$-sided boundary labeling problem, where $1\le k\le 4$, the label locations are located on the $k$ consecutive sides of $R$. 

In this paper, we develop an $O(n^3\log n)$-time algorithm for 2-sided boundary labeling, where the leaders are restricted to have one bend. This  improves the previously best known $O(n^8\log n)$-time algorithm of Kindermann et al. (Algorithmica, 76(1):225--258, 2016). We show the problem is polynomial-time solvable in more general settings such as when the ports are located on more than two sides of $R$, in the presence of obstacles, and even when the objective is to minimize the total number of bends. Our results improve the previous algorithms on boundary labeling with obstacles, as well as provide  the first polynomial-time algorithms for minimizing the total leader length and number of bends for 3- and 4-sided boundary labeling. These results settle a number of open questions on the boundary labeling problems (Wolff, Handbook of Graph Drawing, Chapter 23, Table 23.1, 2014).
\end{abstract}

\section{Introduction}
Labeling problems appear in a variety of   scenarios such as in annotating educational diagrams,  wiring schematics, system manuals, as well as in many engineering and information visualization applications~\cite{handbook,FormannW91,EvansFKKMNV18,BarthNNS16}. The increasing trend of automation in these areas has motivated the research in labeling algorithms. Crossings among the \emph{leaders} (i.e., the lines connecting labels to the sites), number of bends per leader, and the sum of leader lengths are some important aesthetics of a diagram labeling. To achieve clarity and better readability, all these parameters are often preferred to be kept small.  

Many labeling problems are NP-hard~\cite{FormannW91,DBLP:journals/algorithmica/BekosKNS10}. A rich body of research attempts to develop efficient approximation and heuristic  algorithms~\cite{DBLP:journals/isci/Freeman88,hirsh,DBLP:conf/soda/DoddiMMMZ97,DBLP:journals/ior/Zoraster90,DBLP:journals/ior/Zoraster97}, both in the static and the dynamic settings~\cite{BarthNNS16,DBLP:conf/soda/DoddiMMMZ97}. In this paper we examine a well-known variant of the labeling problem called \emph{$b$-bend $k$-sided boundary labeling}, e.g., see Figure~\ref{fig:intro}. The input for this problem is a set of $kn$  sites and $kn$ ports, where the sites lie in the interior of  a rectangle $R$, the ports are located on $k$ consecutive sides of $R$, and each side contains $n$ ports. Both the sites and ports are represented as points. The goal is to decide whether each site can be connected to a unique port using axis-aligned leaders such that the leaders are disjoint, each leader lies strictly inside $R$ and each leader   has  at most $b$ bends. If such a labeling exists, then we compute a labeling that optimizes these labeling aesthetics. We examine two such optimization criteria: one is to minimize the sum of the leader lengths, and the other is to minimize the total number of bends. 

The strict-containment inside $R$, bend restrictions and orthogonal constraints impose certain shapes on the leader. An orthogonal leader containing exactly one bend (resp., 2 bends) is known as a \emph{$po$-leader} (resp., an \emph{$opo$-leader})\footnote{The letters `o' and `p' stands for `orthogonal' and `parallel', respectively. Therefore, an $opo$-leader starts orthogonally at the site, and ends orthogonally at the port.}.

\subparagraph{Related Work:} Boundary labeling has been an active area of research in the last decade, e.g., see the surveys~\cite{wolff,handbook}. The boundary labeling problem was first introduced by Bekos et al.~\cite{DBLP:journals/comgeo/BekosKSW07}. They gave $O(n \log n)$-time algorithms to decide labeling feasibility for the 1-bend 1-sided and opposite 2-sided  models, i.e., the labels are located on two opposite sides of $R$. In addition, they gave an $O(n^2)$-time algorithm that minimizes the total leader length. For the 2-bend 4-sided model, they could test the feasibility in  $O(n \log n)$ time and  reduced the length minimization to a minimum-cost bipartite matching problem. Benkert et al.~\cite{DBLP:journals/jgaa/BenkertHKN09} improved Bekos et al.'s~\cite{DBLP:journals/comgeo/BekosKSW07} result on the 1-bend 1-sided model by devising an $O(n\log n)$-time algorithm for the length minimization. They also considered general cost functions (i.e., beyond Euclidean length), as well as other types of leaders. We refer the reader to ~\cite{DBLP:conf/gis/NollenburgPS10,DBLP:journals/jgaa/BekosCFH0NRS15} for other variants of boundary labeling problem.  

\begin{figure}[pt]
\centering
\includegraphics[width=\textwidth]{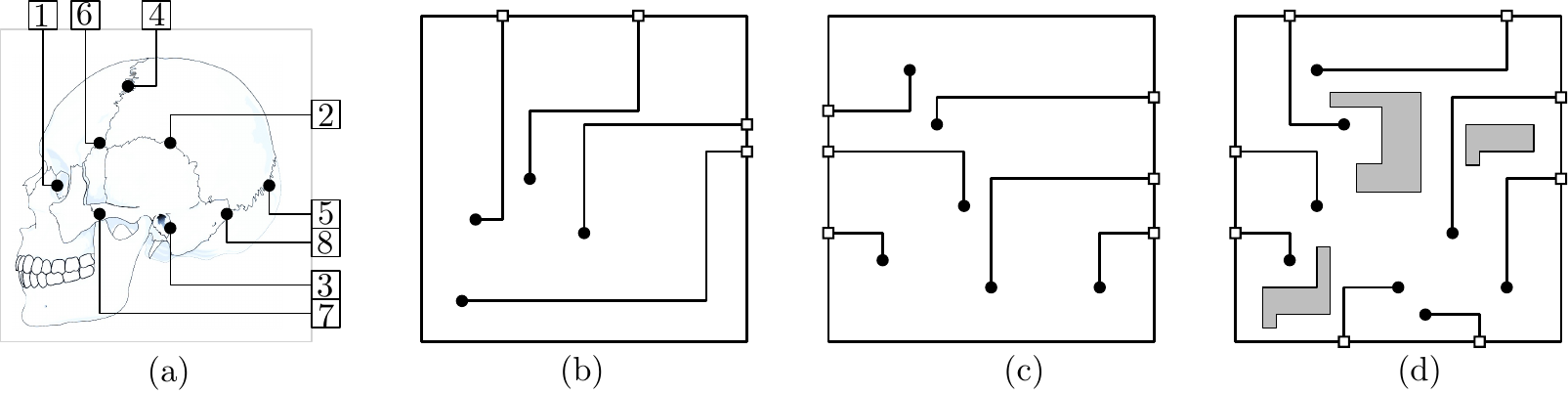}
\caption{(a) A 1-bend 2-sided boundary labeling (i.e., with $po$-leaders). (b) A 2-bend 2-sided boundary labeling   (i.e., with $opo$-leaders). This example does not have a feasible solution with 1-bend leaders. (c) Boundary labeling in 1-bend opposite 2-sided model.  (d) A 1-bend 4-sided boundary labeling in the presence of obstacles.} 
\label{fig:intro}
\end{figure}

The 2-sided model considered by Bekos et al.~\cite{DBLP:journals/comgeo/BekosKSW07} and Benkert et al.~\cite{DBLP:journals/jgaa/BenkertHKN09} is an opposite-sided  model, i.e., ports are placed on two opposite sides of $R$. This model is different from the adjacent 2-sided model, where the labels are always placed on adjacent sides. The adjacent 2-sided model was first considered by Kindermann et al.~\cite{DBLP:journals/algorithmica/KindermannNRS0W16}. For the 1-bend 2-sided model, they gave an $O(n^2)$-time algorithm to check feasibility, and an $O(n^8 \log n)$-time algorithm for total leader length minimization; to our knowledge, this is the fastest algorithm known for the 1-bend 2-sided model. Note that the labeling problem in this model seems surprisingly more difficult than the corresponding opposite 2-sided model (also mentioned by Kindermann et al.~\cite{DBLP:journals/algorithmica/KindermannNRS0W16}). For the 1-bend 3-sided (resp., 4-sided) model, they gave an $O(n^4)$-time (resp., $O(n^9)$-time) algorithm for checking the labeling feasibility, but they were unable to solve the length minimization problem. They posed this as an open question, i.e., can a minimum-length solution for the 3- and 4-sided boundary labeling be computed in polynomial time?  These challenges motivated us to examine the adjacent model in more detail.

Fink and Suri~\cite{DBLP:conf/cccg/0001S16} studied the boundary labeling problem in the presence of \emph{obstacles}. In addition to the set of sites, they allowed a set of orthogonal polygons (equivalently, obstacles) to lie inside $R$. The objective is to minimize   the total leader length with the constraint that the leaders must not intersect the obstacles. They gave polynomial-time algorithms for minimizing the total  leader length in the 1-sided and opposite 2-sided   models, but the running time of these algorithms while using $po$- and $opo$-leaders is fairly high, i.e., $O(n^{4})$, $O(n^{8})$ for the 1-sided model,  and $O(n^{9})$,  $O(n^{21})$ for the opposite 2-sided model. They also examined the case when the leaders have non-uniform lengths and the leader locations can be chosen, which they proved  to be NP-hard. 

A different generalization of boundary labeling considers  sliding ports, i.e., labels are assigned disjoint intervals on the boundary of $R$, and a site can be connected to any point in such an interval. In the 1-sided model, Bekos et al.~\cite{DBLP:journals/comgeo/BekosKSW07} gave an $O(n^2)$-time algorithm that can minimize the total number of bends using $opo$-leader (they did not require the $opo$-leaders to lie strictly insider $R$). They posed an open question to determine the time complexity for the 3- and  4-sided case.  Benkert et al.~\cite{DBLP:journals/jgaa/BenkertHKN09} considered bend minimization with  $po$-leaders. They gave an $O(n^2)$-time algorithm for the 1-sided model, and $O(n^8)$-time algorithm for the opposite 2-sided model. The `Handbook of Graph Drawing'~\cite{handbook} lists a number of open problems related to the minimization of the total number of bends for different variants of boundary labeling.

The 1-, 3- and 4-sided models for the boundary labeling problem are always adjacent models, but a 2-sided model can be either adjacent or opposite. Throughout the paper we will refer to the `opposite' variant as an `opposite 2-sided' model.

\subparagraph{Our Contributions:} We give an algorithm for the 1-bend 2-sided boundary labeling problem that minimizes the total leader length in $O(n^3\log n)$ time (if such a labeling exists). Ours is an adjacent model and uses $po$-leaders, and hence improves the $O(n^8 \log n)$-time algorithm of Kindermann et al.~\cite{DBLP:journals/algorithmica/KindermannNRS0W16}. Since the best known algorithm for the length minimization in the 1-bend opposite 2-sided model takes $O(n^2)$ time~\cite{DBLP:journals/jgaa/BenkertHKN09}, our result raises an intriguing question that whether the adjacent boundary labeling model can further be improved to reach (or, even break) the $O(n^2)$ barrier. 

We show that many variants of the boundary labeling problems can be related to   outerstring graphs, where the  minimization of total leader lengths or bends reduces to an optimization problem in those outerstring  graphs. This idea leads us to the following results:
\begin{enumerate}
\item[-] The first polynomial-time algorithm with a running time of $O(n^6)$ for the 1-bend 3-sided and 4-sided boundary labeling problem that minimize the total leader length. This settles  the time-complexity question posed by Kindermann et al.~\cite{DBLP:journals/algorithmica/KindermannNRS0W16}.
\item[-] Polynomial-time algorithms for minimizing the  total leader length or the total number of bends, even in the presence of obstacles.  Our algorithms work for both $po$- and $opo$-leaders, as well as for  all possible distributions of the ports to the boundary of $R$, i.e., both adjacent and opposite models. The running time for the opposite 2-sided model is $O(n^6)$ for $po$-leaders and $O(n^9)$ for $opo$-leaders; these improve, respectively, the $O(n^9)$- and $O(n^{21})$-time algorithms of Fink and Suri~\cite{DBLP:conf/cccg/0001S16}. This technique can also be applied to the sliding port model, which settles the time-complexity question posed in~\cite{DBLP:journals/comgeo/BekosKSW07, handbook} related to the bend minimization. 
\end{enumerate}

\section{Computing 1-Bend 2-Sided Boundary Labelings}
In this section we give an $O(n^3\log n)$-time algorithm to find a solution to the 1-bend 2-sided boundary labeling problem. Throughout this section, we assume that the sites and ports are in general position, i.e., no axis-aligned straight line passing through a site intersects a port or another site. Consequently, each leader must have exactly one bend. We thus omit the term  `1-bend' in the rest of this section. 

\subsection{Technical Background}
\label{sec:tb}
Let $R(t),R(b),R(l),R(r)$ be the \emph{top, bottom, left and right sides} of $R$, respectively. An \emph{$xy$-separating curve} is an axis-aligned $xy$-monotone polygonal chain that starts at the bottom-left corner of $R$ and ends at the top-right corner of $R$. A 2-sided boundary labeling solution is \emph{$xy$-separated} if there exists an $xy$-separating curve such that the leaders incident to $R(t)$ (resp., $R(r)$) lie on or above (resp., below) the $xy$-separating curve.  

\begin{lemma}[Kindermann et al.~\cite{DBLP:journals/algorithmica/KindermannNRS0W16}]If a 2-sided boundary labeling problem has an affirmative solution with 1-bend leaders, then there exists such an  $xy$-separated solution that minimizes the sum of all leader lengths.
\end{lemma}

Figure~\ref{fig:simple}(a) illustrates an $xy$-separated solution of a 2-sided boundary labeling problem. An $xy$-separated curve is  shown in a light-green. Let $\mathcal{I}$ be an instance of a 2-sided boundary labeling problem. Without loss of generality assume that the ports are distributed along the sides $R(t)$ and $R(r)$.  Let $ports(R(t))$ (resp., $ports(R(r))$) be the set of  ports along $R(t)$ (resp., $R(r)$). A leader is called \emph{inward} if the $90^\circ$ angle formed at its bend point contains the top-right corner of $R$. Otherwise, we call the leader an \emph{outward} leader. The leaders incident to $\ell$ and $\ell'$ in Figure~\ref{fig:simple}(a), are inward and outward leaders, respectively.

\begin{figure}[pt]
\centering
\includegraphics[width=\textwidth]{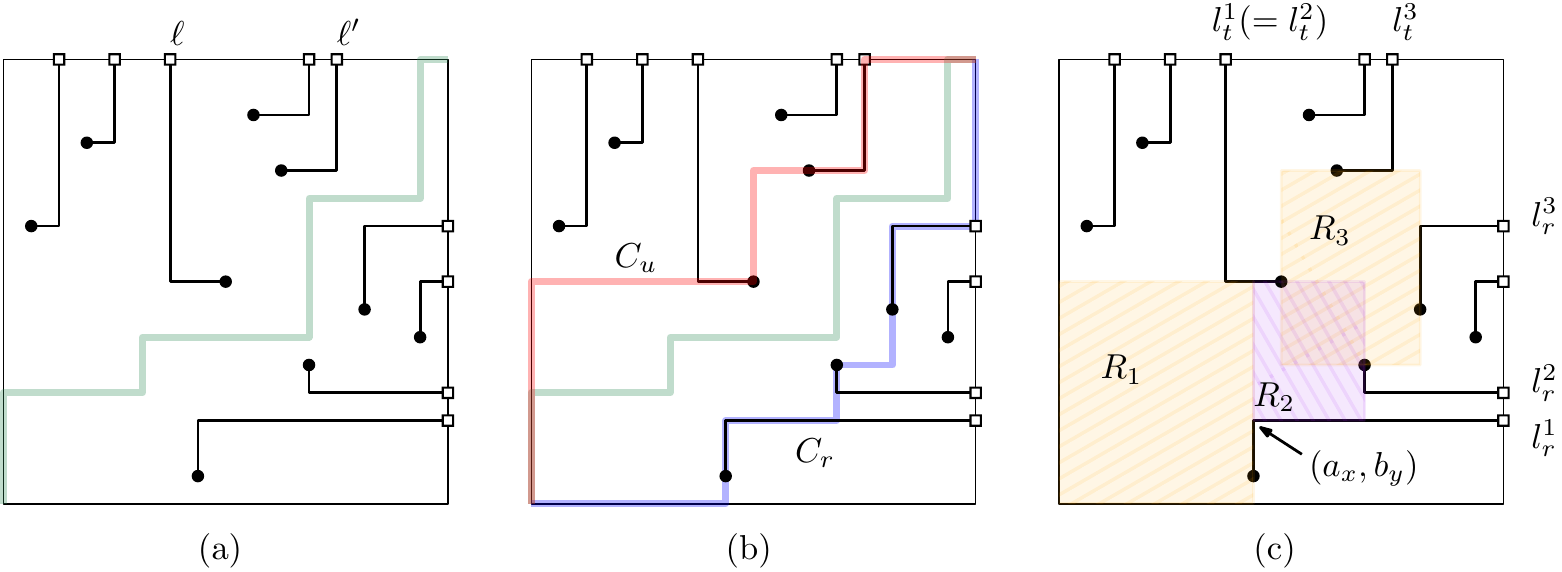}
\caption{(a) An $xy$-separated solution to a 2-sided boundary labeling. The $xy$-separating curve $C$ is shown in light-green. (b) Illustration for the curves $C_u$ and $C_r$. (c) $\mathcal{R}$.}
\label{fig:simple}
\end{figure}

Assume that $\mathcal{I}$ has an affirmative solution $\mathcal{S}$ and let $C$ be a corresponding $xy$-separating curve. Let $up(C)$ be the polygonal region above $C$ bounded by $R(t)$ and $R(l)$. Similarly, let  $right(C)$ be the polygonal region to the right of $C$ bounded by $R(b)$ and $R(r)$. By $C_u$ (resp., $C_r$) we denote the $xy$-separating curve that minimizes the area of $up(C_u)$ (resp., $right(C_r)$), e.g., see Figure~\ref{fig:simple}(b). For a point $p$, let $p_x$ and $p_y$ be its $x$ and $y$-coordinates, respectively. Given $C_u$ and $C_r$, we define a sequence of rectangles $\mathcal{R} = (R_1,R_2,\ldots, R_k)$ as follows:
\begin{enumerate}
\item[-] Each rectangle is a maximal rectangle between $C_u$ and $C_r$.
\item [-] The bottom-left corner of $R_1$ coincides with that of $R$.
\item [-] For $i>1$, we first consider $R_{i-1}$. Since  $R_{i-1}$ is maximal, the top and right sides of $R_{i-1}$ must be determined by a pair of leaders, 
 e.g., see $R_1$ in Figure~\ref{fig:simple}(c). Denote these leaders by $\ell^{i-1}_t$ and $\ell^{i-1}_r$, respectively. Let $a\in \ell^{i-1}_t$ be the rightmost point on the top side of $R_{i-1}$, and let $b\in \ell^{i-1}_r$ be the topmost  point on the right side of $R_{i-1}$. We define $R_i$ to be the maximal empty rectangle with the bottom-left corner at $(a_x,b_y)$ and the sides bounded by $C_u$ and $C_r$. 
\end{enumerate}
 
\subsection{Algorithm}
\label{sec:algo}
The idea of the algorithm is to employ a dynamic programming algorithm based on the idea of finding the optimal rectangle sequence $\mathcal{R}$. Note that for any rectangle $R_j\in \mathcal{R}$, we can think of a subproblem $\lambda(R_j)$ that seeks   a solution  including the leaders $\ell^j_t$ and $\ell^j_r$.   The rectangle  $B(R_j)$ corresponding to this subproblem  is determined by the vertical  and horizontal segments of $\ell^j_t$ and $\ell^j_r$, as illustrated in gray in Figure~\ref{fig:det}(a). It is straightforward to add a dummy rectangle $R_0$ with corresponding leaders $\ell^0_t$ and $\ell^0_r$ such that $\lambda(R_0)$ represents the original  2-sided boundary labeling problem, e.g., see  
 Figure~\ref{fig:det}(b).

\begin{figure}[pt]
\centering
\includegraphics[width=\textwidth]{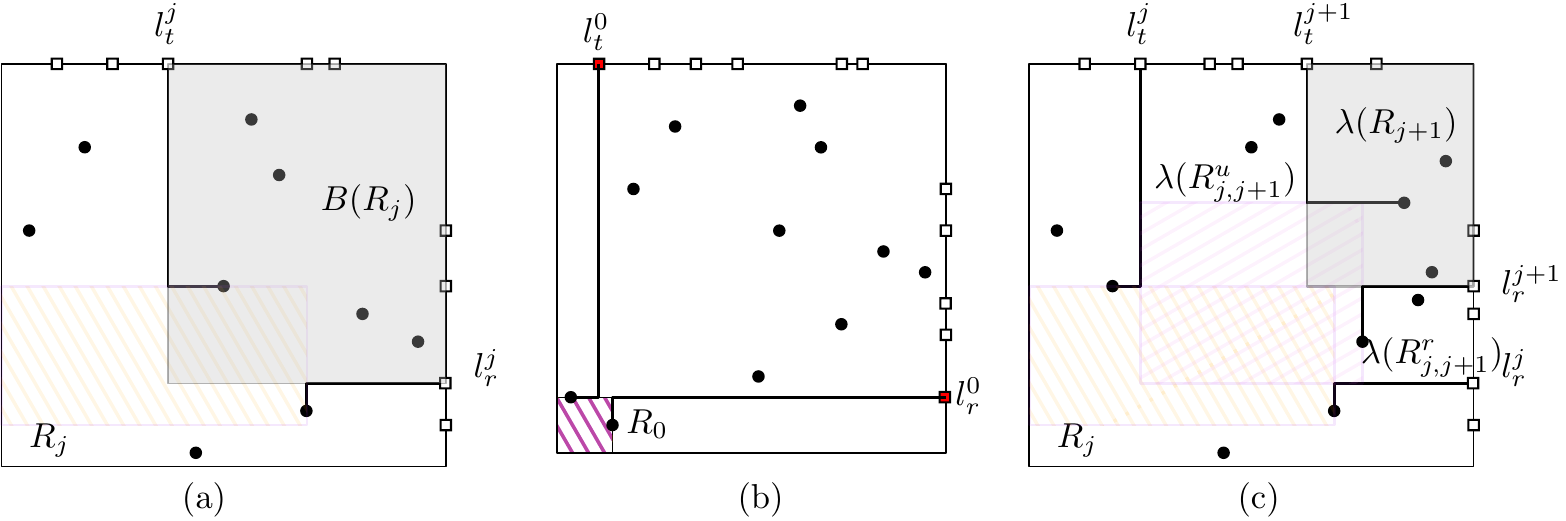}
\caption{Illustration for the dynamic programming algorithm.} 
\label{fig:det}
\end{figure} 

Given $R_j$, we try to find $R_{j+1}$ by checking all possible candidate rectangles. For convenience, we defer the details of finding all candidate rectangles, and focus on the  computation of the solution cost (sum of leader length) assuming that we have found $R_{j+1}$. Figure~\ref{fig:det}(c) illustrates such a scenario. Let $R^t_{j,j+1}$ be the region bounded by the lines determined by the vertical segments  of $\ell^j_t$ and $\ell^{j+1}_t$, horizontal segment of  $\ell^j_t$, and $R(t)$. Define $R^r_{j,j+1}$ symmetrically, e.g., see the top of Figure~\ref{fig:finedet}(i).  Observe that $\lambda(R^t_{j,j+1})$ is a 1-sided boundary labeling problem with leaders $\ell^j_t$ and $\ell^{j+1}_t$. In other words, since $R_{j+1}$ is an empty rectangle, all the ports between $\ell^j_t$ and $\ell^{j+1}_t$ must be connected to some site interior to $R^t_{j,j+1}$. We define $\lambda(R^r_{j,j+1})$ symmetrically. It is now straightforward to express the solution of $\lambda(R_j)$ in terms of the solutions of  $\lambda(R^t_{j,j+1}),\lambda(R^r_{j,j+1})$, and $\lambda(R_{j+1})$.

For any leader $l$, we denote its length by $|l|$. Let $|\lambda(R_j)|$ be the sum of the leader lengths in an optimal solution of $\lambda(R_j)$ (excluding the lengths of $\ell^j_t$ and $\ell^j_r$). Let $ports(B(R_j))$ and $sites(B(R_j))$ be the number of ports and sites interior to $B(R_j)$, excluding those that are incident to $\ell^j_t$ and $\ell^j_r$. We now have the following recursive formula, where $\mathcal{C}$ denotes the set of candidate rectangles. 
\begin{equation*}
  |\lambda(R_j)|=\begin{cases}
    \infty, & \text{if $ports(B(R_j)) \not= sites(B(R_j))$}.\\
    (|\ell^j_t| + |\ell^j_r|)+ \\\min\limits_{R_{j+1}\in \mathcal{C}}\{ |\lambda(R^t_{j,j+1})| {+} |\lambda(R^r_{j,j+1})| {+} |\lambda(R_{j+1})|\}, & \text{otherwise}.
  \end{cases} 
\end{equation*}

\subparagraph{Finding Candidate Rectangles:}
Given a rectangle $R_j$, we now describe how to find a set of candidate rectangles that must include $R_{j+1}$. Recall that we can compute the bottom-left corner $(a_x,b_y)$ of $R_{j+1}$ from $R_j$.    Figures~\ref{fig:finedet}(a)--(d) illustrate the scenarios where $\ell^j_t$ and $\ell^j_r$ are inward. The point $(a_x,b_y)$  is marked with a cross. We claim that either the top or the right side of $R_{j+1}$ must contain a site (Lemma~\ref{lem:site}). We will use the following result of Benkert et al.~\cite{DBLP:journals/jgaa/BenkertHKN09} to prove Lemma~\ref{lem:site}. 

\begin{figure}
\centering
\includegraphics[width=\textwidth]{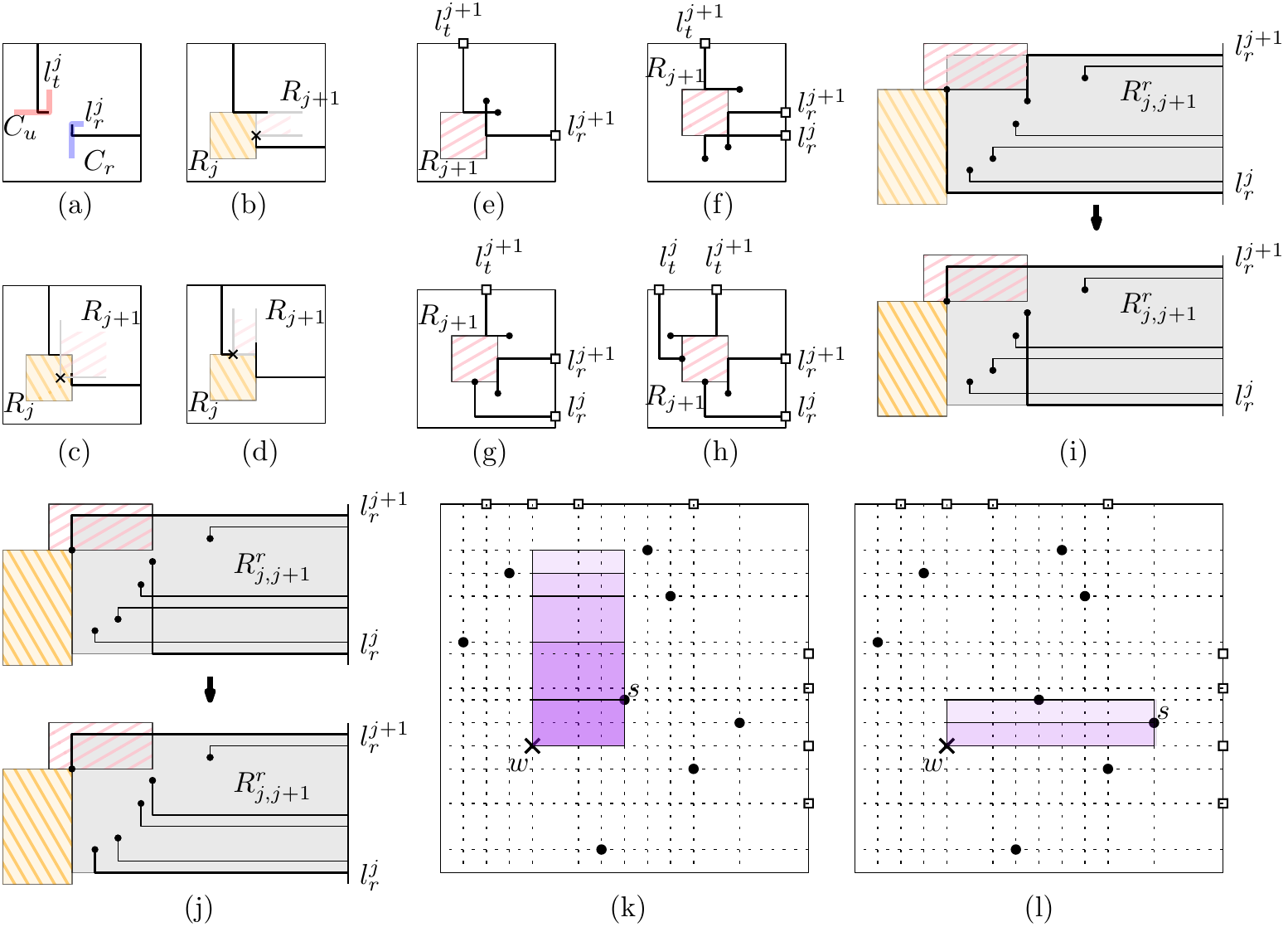}
\caption{Illustration for (a)--(d)  $(a_x,b_y)$,  (e)--(j)  Lemma~\ref{lem:site}, and (k)--(l)  candidate rectangles.}
\label{fig:finedet}
\end{figure}

\begin{lemma}[Benkert et al.~\cite{DBLP:journals/jgaa/BenkertHKN09}]\label{lem:benkert}
For any solution $S$ to a 1-bend 1-sided boundary labeling problem  that minimizes the total leader length (possibly with crossings), there exists a crossing free labeling with the total leader length at most the total leader length of $S$.
\end{lemma}

\begin{lemma}\label{lem:site}
Either the top or the right side of $R_{j+1}$ must contain a site. 
\end{lemma}
\begin{proof}
Suppose for a contradiction that neither the top nor the right side of $R_{j+1}$ contains a site. We now consider a few cases. 

\textbf{Case 1 (both $\ell^{j+1}_t$  and $\ell^{j+1}_r$  are inward):} In this case the leaders $\ell^{j+1}_t$ and $\ell^{j+1}_r$ must intersect (see Figure~\ref{fig:finedet}(e)), which contradicts that the underlying solution is crossing-free. 

\textbf{Case 2 ($\ell^{j+1}_t$ is inward and $\ell^{j+1}_r$  is outward):} If $\ell^j_r$ is outward, then it must intersect $\ell^{j+1}_r$ (see Figure~\ref{fig:finedet}(f)). Therefore, the leader $\ell^j_r$ must be inward, as illustrated in  Figure~\ref{fig:finedet}(g). Note that by our general position assumption, the `$y$-intervals' determined by the vertical segments of $\ell^{j}_r$  and $\ell^{j+1}_r$ must overlap. Consequently, by swapping the site assignments, we can obtain a solution  
(possibly with crossings) with strictly smaller total leader length. Figure~\ref{fig:finedet}(i) illustrates such a scenario. 
By Lemma~\ref{lem:benkert}, we can replace this labeling of  $\lambda(R_{j,j+1})$ with a crossing free labeling that lies inside $R^r_{j,j+1}$ and does not increase the total leader length, e.g., see  Figure~\ref{fig:finedet}(j). Note that the total leader length of the resulting solution would be strictly smaller, contradicting that the current solution is optimal. 

\textbf{Case 3 ($\ell^{j+1}_t$ is outward and $\ell^{j+1}_r$  is  inward):} This case is symmetric to Case 2.

\textbf{Case 4 (both $\ell^{j+1}_t$ and $\ell^{j+1}_r$ are outward):} 
 We can process this case in the same way as we did in Case 2.
\end{proof}

Recall that we know the bottom-left point $w$ of $R_{j+1}$. We first assume that the right side of $R_{j+1}$ contains a site. For every site $s$ with $s_x> w_x$ and $s_y>w_y$, we consider all possible empty rectangles with bottom-left corner $w$, right side passing through $s$ and the top side determined by a horizontal line passing through a site above $s$. Figures~\ref{fig:finedet}(k)--(l) illustrate the candidate rectangles for the bottom left point $w$. We then assume that the top side of $R_{j+1}$ contains a site, and find the candidate rectangles symmetrically. We can now  obtain an upper bound on the distinct candidate rectangles.
\begin{lemma}\label{lem:cr}
The overall number of distinct candidate rectangles  
 is $O(n^3)$. 
\end{lemma}
\begin{proof}
For a particular bottom-left corner $w$, it may initially appear that there are  $O(n^2)$ possible candidate rectangles to explore. But we can prove an $O(n)$ upper bound, as follows.

Let $D$ be the region dominated by $w$; i.e., for each point $q\in D$, the $x$ and $y$-coordinates of $q$ are at least as large as those of $w$. Let $S = \{s_1,s_2,\ldots,s_k\}$ be the  set of sites in $D$ (ordered by increasing $y$-coordinates) such that no site in $S$ is dominated by any other site in $D$ (except possibly for $w$). We may assume without loss of generality (see Lemma~\ref{lem:site}) that the right side of $R_{j+1}$  contains a site. Since the proper interior of the rectangle is empty, for each $s_i$, where $1\le i\le k$, we only need to consider a set of heights $H(s_i)$ that lie between $s_i$ and $s_{i+1}$ (or, between $s_i$ and $R(t)$ when $i=k$).   For every pair of sites $\{s,s'\}\in S$, we have the property that neither $s$ nor $s'$ dominates the other. Therefore, we have $H(s)\cap H(s') = \phi$,  $\sum_i{H(s_i)} = O(n)$, and thus a linear number of candidate rectangles for $w$.

The number of possible intersections (i.e., bottom-left corners) among the horizontal and vertical lines passing through the ports and sites is $O(n^2)$. Therefore, the number of distinct candidate rectangles that may appear over the run of the algorithm is $O(n^3)$.
\end{proof}

\subparagraph{Data-structures and Time Complexity:} If we use an $O(n^2)\times O(n^2)$ dynamic programming table and compute each entry by checking $O(n)$ candidate rectangles, then we need at least $O(n^5)$ time.  To improve the running time to $O(n^3 \log n)$, we preprocess the input.  
 For every possible  matching of a pair of ports  (on the same side of $R$) to a pair of sites, we compute and store the solution to the  corresponding 1-sided boundary labeling problem. Since there are $O(n^4)$ such 1-sided problems, and each of them can be answered in $O(n\log n)$ time~\cite{DBLP:journals/jgaa/BenkertHKN09},  this takes $O(n^5 \log n)$ time. 
 We first show how to reduce this preprocessing time to $O(n^3\log n)$.

Consider a subproblem $\lambda(R^t_{j,j+1})$. Such a problem can easily be expressed by the ports and sites incident to $\ell^j_t$ and $\ell^{j+1}_t$.  Here we encode $\lambda(R^t_{j,j+1})$ in a slightly different way. We use the parameters $p,p',\alpha ,\beta$, where $p,p'$ are the ports incident to $\ell^j_t$ and $\ell^{j+1}_t$, $\alpha$ is either $\infty$ or the $y$-coordinate of a site,  and $\beta$ is the `type' of $\lambda(R^t_{j,j+1})$. We will express $\lambda(R^t_{j,j+1})$ as $S(p,p',\alpha,\beta)$. In the following we  describe the details of $S(p,p',\alpha,\beta)$.

Note that to solve $\lambda(R^t_{j,j+1})$ affirmatively, we need exactly as many free sites as the number of ports between  $p$ and $p'$. Thus for any subproblem, if the number of free sites and free ports interior to $R^t_{j,j+1}$ do not match, then we can immediately return a negative answer. We assume that the points and ports are stored in an orthogonal range counting data structure (with $O(n\log n)$-time preprocessing) such  that given an orthogonal rectangle, one can report the number of ports and points interior to the rectangle in $O(\log n)$ time~\cite{Berg08}.  We only focus on those instances that have the same  number of  free sites and ports, and express them in the form  $S(p,p',\alpha,\beta)$. 

Let $s,s'$ be the sites that are incident to $\ell^j_t$ and $\ell^{j+1}_t$, respectively. By the property of the optimal solution, we may assume that $s_y< s'_y$. We define $\lambda(R^t_{j,j+1})$  as having type 0, 1, 2 or 3 depending on whether $s,s'$ belongs to $R^t_{j,j+1}$ or not. 

\textbf{Type 1 (both $s,s'$ are outside $R^t_{j,j+1}$):} In this case the rectangle determined by the bend points of $\ell^{j}_t$ and $\ell^{j+1}_t$  must be empty (i.e., the gray region in Figure~\ref{options}(a)). We set $\alpha$ to be $\infty$, and $\beta$ to be 1. During the algorithm execution, if $\lambda(R^t_{j,j+1})$ is of Type 1, then we will seek   a solution to  $S(p,p',\infty,1)$. 

Note that for any instance of the form  $S(p,p',\infty,1)$, we can determine in $O(1)$ time\footnote[2]{It is straightforward to preprocess the ports and sites in $O(n^3)$ time in a data structure to answer such queries in $O(1)$ time.}     the point $s''$  such that the rectangle $B$ determined by $p,p',s''$ contains an equal number of free ports and sites. Note that the solution to  the labeling problem inside $B$ will be equivalent to that of $\lambda(R^t_{j,j+1})$. We will precompute the solutions of $S(p,p',\infty,1)$ so that $\lambda(R^t_{j,j+1})$ can be answered in $O(1)$ time by a table look-up. This general idea of answering a problem $\lambda(\cdot)$ using $S(\cdot)$ applies also to the other types, i.e., Types 2,3 and 4.

\textbf{Type 2 (both $s,s'$ are inside $R^t_{j,j+1}$):} In this case the rectangle determined by the bend point of $\ell^{j+1}_t$ and $s$ must be empty (Figures~\ref{options}(b)--(c)). We thus set $\alpha$ to be $y(s')$, and $\beta$ to be 2. Observe now that given $S(p,p',\alpha,2)$, we can find both $s$ and $s'$ in $O(1)$ time\footnotemark[2] by counting the number of ports between $p$ and $p'$, and using $\alpha$.  

\textbf{Type 3 ($s\in R^t_{j,j+1}$ and $s'\not\in R^t_{j,j+1}$):} In this case the rectangle determined by the bend point of $\ell^{j+1}_t$ and $s$ must be empty (Figure~\ref{options}(d)). We thus set $\alpha$ to be $y(s')$, and $\beta$ to be 3. Given $S(p,p',\alpha,3)$, we can recover $s$ and $s'$ using  the range counting data structures\footnotemark[2]. The same argument holds even when $s_x>p'_x$.  

\textbf{Type 4 ($s\not\in R^t_{j,j+1}$ and $s'\in R^t_{j,j+1}$):} In this case the rectangle determined by the bend points of $\ell^j_t$ and $\ell^{j+1}_t$ must be empty (Figure~\ref{options}(e)). We thus set $\alpha$ to be $y(s')$, and $\beta$ to be 4. Given $S(p,p',\alpha,4)$, we can recover $s'$ using $\alpha$. Here we do not need to find $s$ since the solution must lie inside the rectangle determined by $p,p'$ and $s'$.

\begin{figure}
\centering
\includegraphics[width=\textwidth]{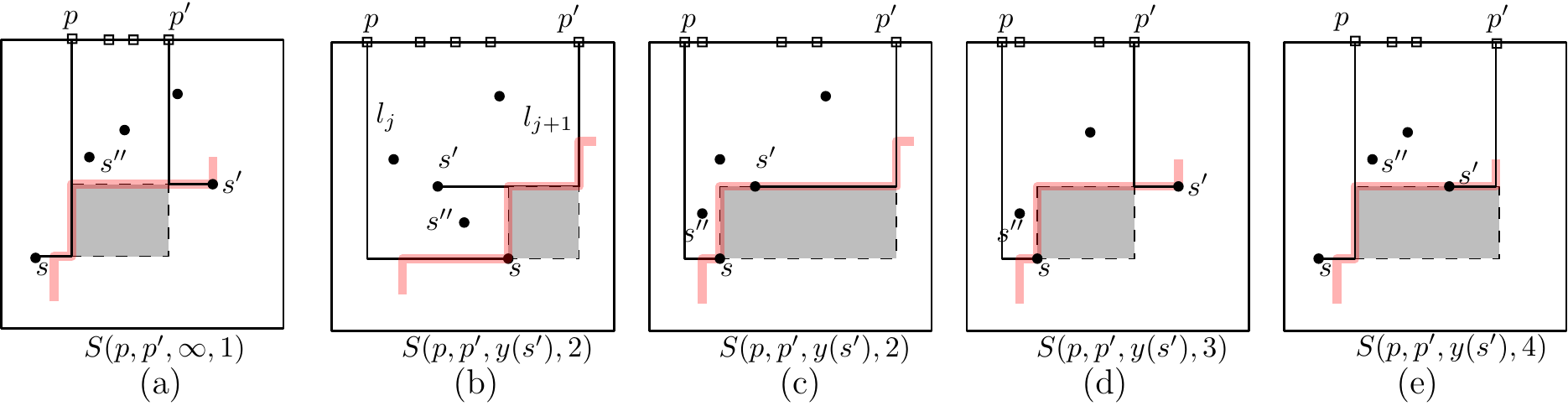}
\caption{(a)--(e) Illustration for different Types of subproblems. }
\label{options}
\end{figure}

\begin{lemma}
\label{lem:bottleneck} 
The solution to the problems $S(p,p',\alpha,\beta)$ can be computed in $O(n^3 \log n)$ time.
\end{lemma}
\begin{proof}
First observe that there are $O(n)$ possible choices for each of $p,p',\alpha$, and a constant number of choices for $\beta$. Therefore, we have at most $O(n^3)$ subproblems. 

We can employ a dynamic programming to compute the solution to these problems. The idea is to select the bottommost free point $s''$ and connect it to a port $p''$ between $p$ and $p'$. This splits the problem into two subproblems, which can again be expressed in the form $S(p,p',\alpha,\beta)$.  Since the number of choices for $p''$ is at most $n$, we can compute an entry of the dynamic programming in $O(n)$ time\footnote{We do not recompute the subproblems and perform a table look-up.}. Since the number of entries is $O(n^3)$, the running time is bounded by $O(n^4)$. 

We now show how to compute an entry using only a logarithmic number of queries, which would imply a running time of $O(n^3\log n)$. Note that the problems of Types 1 and 4 can essentially be thought of as an 1-sided problem, and we can precompute the solution of all such cases in $O(n^3\log n)$ time using Benkert et. al.'s~\cite{DBLP:journals/jgaa/BenkertHKN09} algorithm\footnote{For each pair of vertical lines passing through ports, we apply the $O(n\log n)$ algorithm to solve the 1-bend 1-sided problem using the topmost $k$ sites (if exists), where $k$ is the number of free ports between $p$ and $p'$.}.  We may thus assume without loss of generality that the problem is of Type 2 or 3. 

Let $\Gamma_e$ be the rectangle determined by $s$ and the bend point of the leader of $p'$, i.e., the shaded rectangles in Figure~\ref{options}. Note that no leader in the solution of $\lambda(R^t_{j,j+1})$ would intersect   $\Gamma_e$. The union of the horizontal line of the leader of $p'$ and the top  side of $\Gamma_e$ defines a \emph{blocker}, i.e., a segment $ab$ that must not be intersected by any leader in the solution. Figure~\ref{fig:definition}(a)--(b) illustrate a blocker in a thick line segment.  
\begin{figure}[h]
\centering
\includegraphics[width=\textwidth]{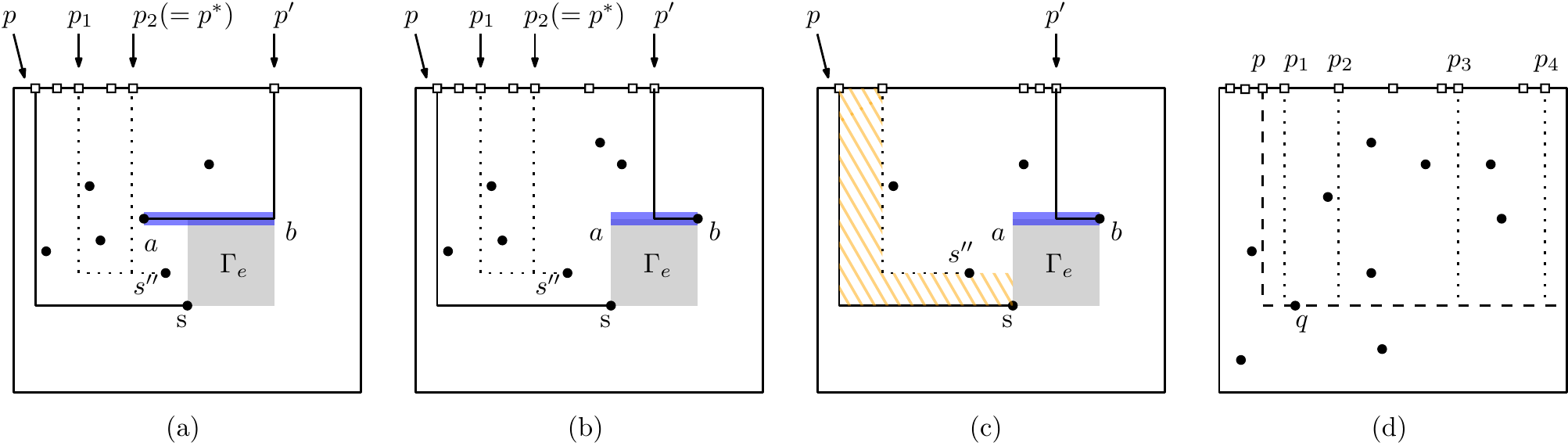}
\caption{Illustration for $\Gamma_e$, blocker, good ports and $p^*$. (a) A scenario when the blocker coincide with the horizontal segment of $\ell^{j+1}_t$. (b) The blocker is determined by the top boundary of $\Gamma_e$. (c) Illustration for Case 1. (d) Precomputation of the list of good ports.}
\label{fig:definition}
\end{figure} 

For any pair of ports $p_i,p_j$, let $S(p_i,p_j)$ be the stripe lying between the leaders incident to the ports $p_i,p_j$ and bounded by $C^u$. We call $S(p_i,p_j)$ a \emph{balanced stripe} if and only if the number of free sites inside the stripe is equal to the number of free  labels between $p_i,p_j$. We call a port $\bar{p}$ (between $p$ and $p'$) a \emph{good port} if $S(p,\bar{p})$ is balanced (assuming that $\bar{p}$ is connected to $s''$), e.g., the ports $p_1$ and $p_2$ in Figures~\ref{fig:definition}(a)--(b). 

Let  $p_1,\ldots,p_k(=p^*)$ be the good ports ordered from left to right, and let $P(p_i,p')$, where $1\le i\le k$, denote the 1-sided boundary labeling problem determined by the leader of $\ell^{j+1}_t$ and the leader connecting $p_i$ to $s''$. By $p^*$ we denote the rightmost good port to the left of the blocker such that the $P(p^*,p')$ has a feasible solution. We now consider two scenarios depending on whether $p^*$ exists or not. 

\textbf{Case 1 ($p^*$ is the port immediately to the right of $p^*$):} In this case we must connect $s''$ to $p^*$, and the region to the left of this leader  must be empty, e.g., see the region with falling pattern in  Figure~\ref{fig:definition}(c). Suppose for a contradiction that $s''$ is connected to some port $\bar{p}$ in the optimal solution such that $\bar{p}_x> p^*_x$. Since $s''$ is the bottommost site, $S(p,\bar{p})$ must be balanced, and hence $p^*_x \ge \bar{p}_x$, a contradiction. Therefore, in this case we need only one table look-up to fill a dynamic programming table entry.  

\textbf{Case 2 (Otherwise):} In this case there can be one or more good  ports. We  prove the following property in Section~\ref{sec:separation}.
 
\begin{enumerate}
\item[]\textbf{Good-Port-Separation:} If $P(p_i,p')$ has a feasible solution, then each of $P(p_1,p'), \ldots,$ $ P(p_{i-1},p')$ has a feasible solution. In addition, there exists an optimal solution of $\lambda(R^t_{j,j+1})$ such that no leader  crosses the vertical line through $p^*$ (except possibly the leaders of $p$ and $p^*$).  
\end{enumerate}
 
By the Good-Port-Separation property, the stripe $S(p,p^*)$ represents a 1-sided boundary labeling problem that can be solved using Benkert et al.'s algorithm~\cite{DBLP:journals/jgaa/BenkertHKN09}. Let $p_{\text{left}}$ and $p_{\text{right}}$ be the closest good ports to the left and right of $s''$,  respectively. We then have to perform at most two table look-up, one for  connecting  $s''$ to $p_{\text{left}}$ and the other for  to $p_{\text{right}}$. Note that if $p^*$ lies to the left of $s''$, then $p_{\text{right}}$ may not exist. 

We now show how to test whether $p^*$ exists, and if so, then how to find $p_{\text{left}}$ and $p_{\text{right}}$; all in  $O(\log n)$ time. We rely on some preprocessing. For every site $q$ (which will play the role of $s''$) and port $p$, we precompute a list of good ports (sorted in increasing order of $x$-coordinates) to the right of the vertical line through $p$, as illustrated in Figure~\ref{fig:definition}(d).  By the   Good-Port Separation  property, given $R_{j,j+1}$, we can perform a binary search on the list of $p'$ to decide whether $p^*$ exists. If $p^*$ exists, then another binary search for the $x$ coordinate of $s''$ would locate  $p_{\text{left}}$ and $p_{\text{right}}$.
\end{proof}

\begin{theorem}
\label{thm:main}
Given a 1-bend 2-sided boundary labeling problem with $O(n)$ sites and labels, one can find a labeling  (if exists) that minimizes the total leader length in $O(n^3\log n)$-time.  
\end{theorem}
\begin{proof}
Every subproblem $\lambda(R_j)$ can be defined by a pair of leaders, and hence we can define an $O(n^2)\times O(n^2)$ table $T$ to store the solutions to the subproblems. To compute an entry of the table $T$, we look at a set of candidate rectangles with two nice properties. First, all these rectangles have the same  bottom-left corner, and second, none of these rectangles can be a candidate rectangle for any other entry of $T$. Therefore, the number of `candidate rectangle queries' to fill all the entries of $T$ is  bounded asymptotically by the number of distinct candidate rectangles, which is $O(n^3)$ (by Lemma~\ref{lem:cr}). Since we do not recompute solutions, and the table look-up takes $O(1)$ time, the total running time is bounded by $O(n^4)$, which dominates the preprocessing time. 

Observe that the complexity $O(n^4)$ comes from considering all possible pairs of strings, whereas only $O(n^3)$ options are relevant (by Lemma~\ref{lem:cr}). Therefore, instead of a table, we can keep the relevant  entries in a dynamic binary search tree, which increases the cost for solution look-up to $O(\log n)$, but limits the time for both the memory initialization and look-up queries to $O(n^3 \log n)$. Thus the total running time improves to  $O(n^3 \log n)$. 
\end{proof}

\section{Good-Port-Separation Property}
\label{sec:separation}
In this section, we prove the Good-Port-Separation property.

\begin{lemma}[Good-Port-Separation]
\label{lem:separation} 
 If $P(p_i,p')$ has a feasible solution, then each of $P(p_1,p'), \ldots,$ $ P(p_{i-1},p')$ has a feasible solution. In addition,  
 there exists an optimal solution of $\lambda(R^t_{j,j+1})$ such that no leader  crosses the vertical line through $p^*$ (except possibly the leaders of $p$ and $p^*$).
\end{lemma}
\begin{proof} Suppose for a contradiction that $P(p_a,p')$ is not feasible, where $a<i$. We now construct a feasible solution as follows: We first construct a 1-sided solution of $S(p,p_a)$ (such a solution exists since the stripe is balanced~\cite{DBLP:journals/jgaa/BenkertHKN09}), e.g., see Figure~\ref{fig:split}(a). We know that $S(p_i,p')$ has a feasible  solution, where by definition, $p_i$ is connected to the bottommost point $s''$. We take such a solution of $P(p_i,p')$ and remove the leader connecting $p_i$ and $s''$. We then connect $s''$ to $p_a$, and the stripe $S(p_a,p_i)$ (including $p_i$) becomes balanced, e.g., see Figure~\ref{fig:split}(b). Consequently, we can construct a feasible labeling for $S(p_a,p_i)$, which contradicts that $P(p_a,p')$ is not feasible. 
 
We now show that there exists an optimal solution of  $\lambda(R^t_{j,j+1})$ such that no leader would cross the vertical line $L$ through $p^*$ (except possibly the leaders of $p^*$ and $p$). Assume that there exists an optimal solution $S$, where some of the leaders intersect $L$. Note that it would suffice to transform this solution into another solution $S^*$ that respects the separation by $L$, and does not increase the sum of leader length, i.e., $cost(S^*)\le cost(S)$. We compute the transformation in two phases. In the first phase, we transform $S$ into some solution $S'$ with $cost(S')\le cost(S)$ that respects the separation, but possibly contains crossings among leaders. In the next phase we transform $S'$ into some solution $S^*$ with $cost(S^*)\le cost(S')$, which is crossing free and respects the separation.

We first split $R^t_{j,j+1}$ into some regions as follows: Let $ab$ be the blocker segment. 
 Define $\Gamma_1$ to be the rectangle determined by $a,p'$, $\Gamma_2$ be the rectangle determined by $a,p^*$, and $\Gamma_3$ be the rectangle bounded by the four lines determined by the horizontal segments of   $\ell^j_t$ and $\ell^{j+1}_t$, and the vertical lines through $p^*$ and $p'$. Figure~\ref{fig:split}(c) illustrates $\Gamma_1, \Gamma_2$ and $\Gamma_3$. We now consider the following scenarios. 

\begin{figure}[h]
\centering
\includegraphics[width=\textwidth]{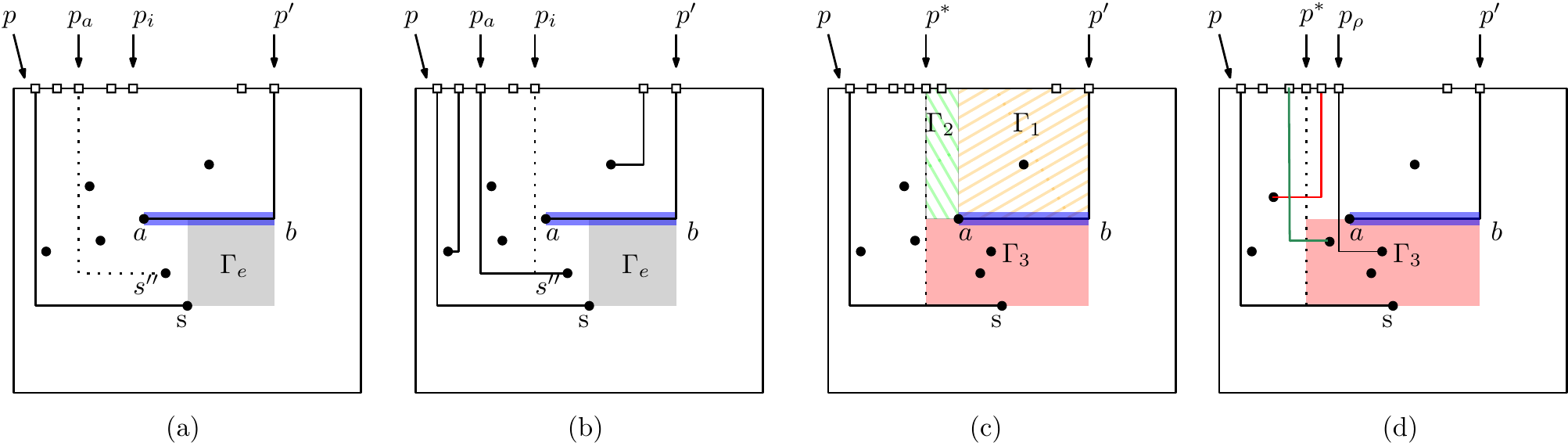}
\caption{Illustration for Lemma~\ref{lem:separation}. (a)--(b) Construction of a feasible solution for $P(p_a,p')$. (c)  $\Gamma_1, \Gamma_2$ and $\Gamma_3$. (d) Case 1.}
\label{fig:split}
\end{figure} 

\textbf{Case 1:} First consider the case when there exists a site in $\Gamma_3$ whose leader does not intersect $L$. Let $p_\rho$ be the port incident to this leader. Note that among the ports that lie to the right of $L$, only the leaders incident to the ports between $p^*$  and $p_\rho$ may properly\footnote{The leader incident to $p^*$ intersects $L$, but we do not consider that as a proper intersection.} intersect $L$ (e.g., see Figure~\ref{fig:split}(d)). Since $S(p,p^*)$ is balanced, if there is a set of $k$ such leaders intersecting $L$, then there must be another set of $k$ leaders that are incident to the ports to the left of $p^*$ and intersect $L$. Figure~\ref{fig:split}(d) illustrates these leaders in red and green, respectively. We now arbitrarily compute $k$ pairs, taking one leader from each set, and swap their sites. Such a swap may create edge crossings (e.g., see Figures~\ref{fig:transformation}(a)--(b)). However, they do not increase the total leader length, as illustrated in Figures~\ref{fig:transformation}(c)--(f). This yields the solution $S'$ that respects the separation by $L$. 
 The solution to the left and to the right of $L$ now can independently be modified to remove the crossing by using  a result of Benkert et al.~\cite[Lemma 1]{DBLP:journals/jgaa/BenkertHKN09}.
\begin{figure}[h]
\centering
\includegraphics[width=\textwidth]{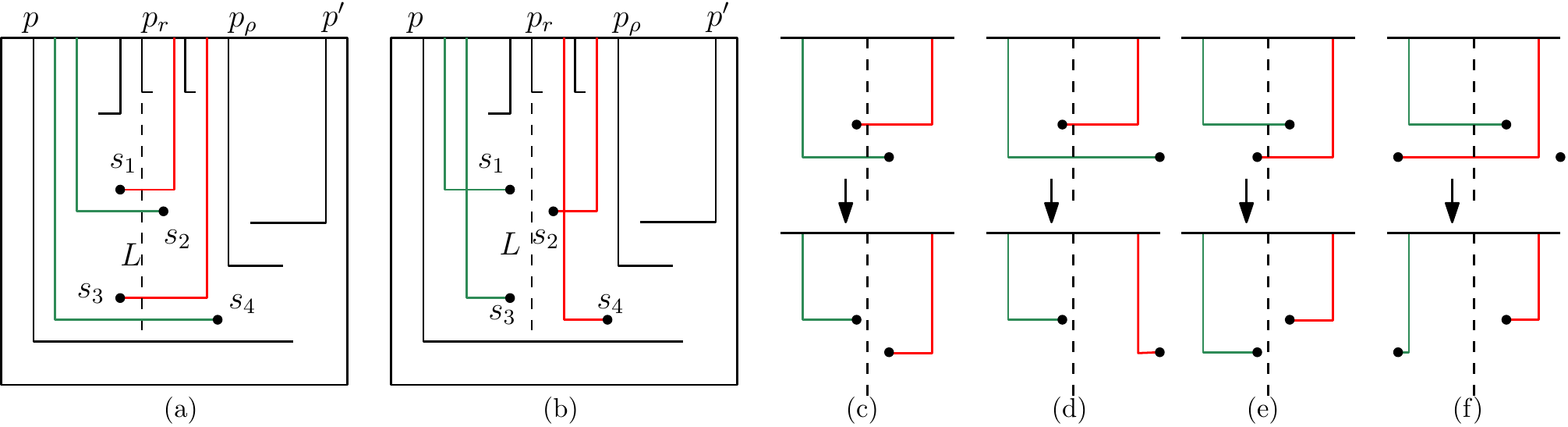}
\caption{(a)--(b) Illustration for the construction in Case 1. (c)--(f) Swap of the site assignments does not increase the total leader length.}
\label{fig:transformation}
\end{figure}

\begin{figure}
\centering
\includegraphics[width=\textwidth]{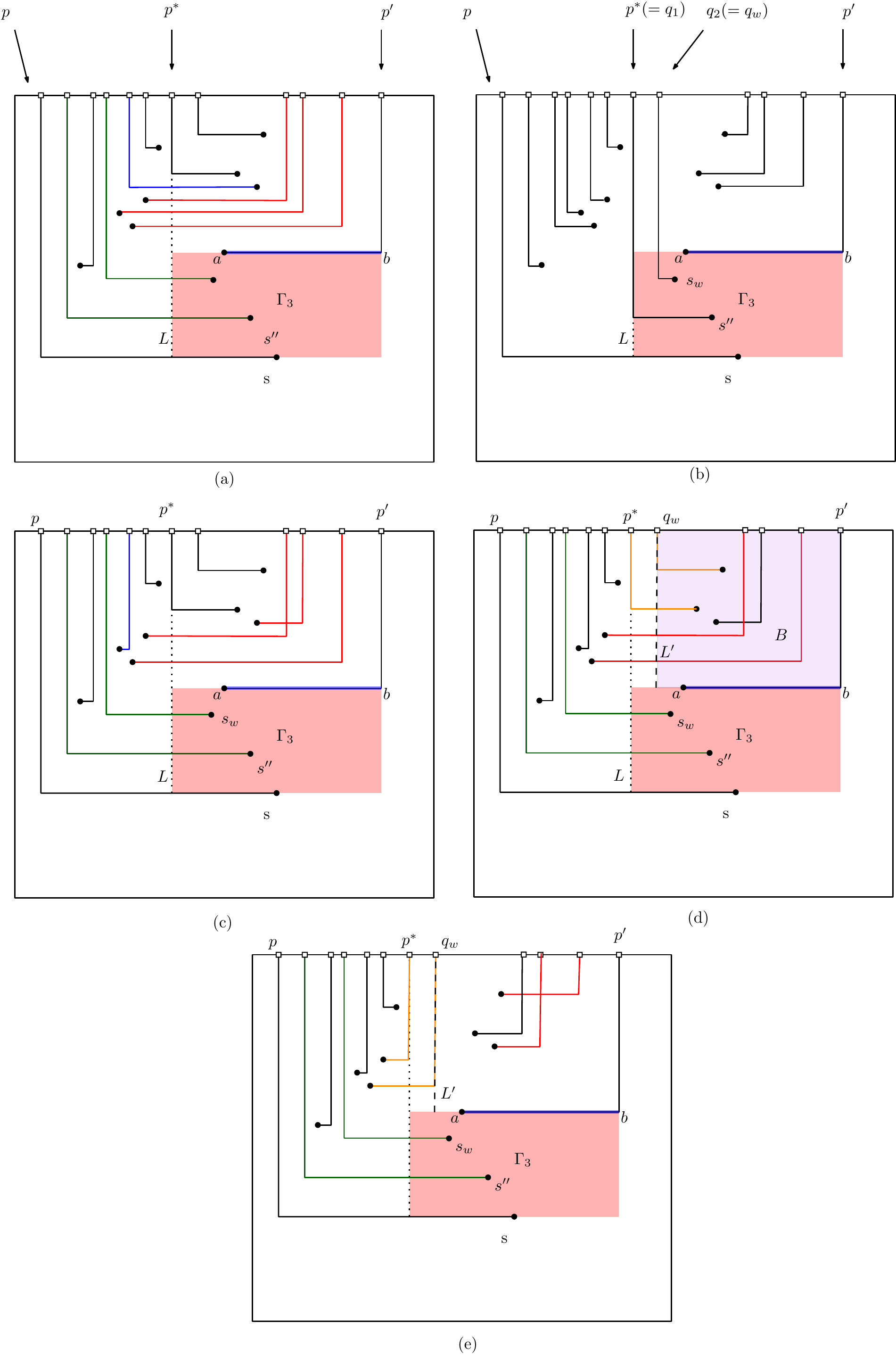}
\caption{Illustration for the proof of Lemma~\ref{lem:separation}.}
\label{fig:split2}
\end{figure}

\textbf{Case 2:} Consider now the case when all the  sites of $\Gamma_3$   intersect $L$. Let $w$ be the number of sites interior to $\Gamma_3$, and let $\mathcal{L}_w$ be the set of leaders incident to these sites (the green leaders in Figure~\ref{fig:split2}(a)). Recall that  $P(p^*,p')$ has a feasible solution. Let $q_1,\ldots,q_w$ be the ports from left to right in $\Gamma_2$ that are connected to the sites in $\Gamma_3$ in an optimal solution of  $P(p^*,p')$, e.g., see  Figure~\ref{fig:split2}(b). Let $s_w$ be the port incident to $q_w$.

Since the leaders in $\mathcal{L}_w$ are crossing $L$, any other leader crossing $L$ and incident to a port to the right of $L$, must lie above the line determined by the blocker. Assume that there are $w'\ge w$ such leaders and define $\mathcal{L}_{w'}$ to be the set that consists of these leaders  (i.e., the red leaders in  Figure~\ref{fig:split2}(a)). 
 Consequently, we must have another set of $w''= (w'-w)$ leaders that are incident to the ports lying to the left of $L$ and intersecting $L$ (i.e., the blue leader in  Figure~\ref{fig:split2}(a)). Denote these set of leaders by $\mathcal{L}_{w''}$. Since the leaders in $\mathcal{L}_w$ cover all the points in $\Gamma_3$, each leader in $\mathcal{L}_{w''}$ must lie above the line determined by the blocker.

We now construct $w''$ pairs of leaders, taking one leader from $\mathcal{L}_{w''}$ and one from $\mathcal{L}_{w'}$, and for each pair, we swap their sites. Similar to Case 1, these swaps do not increase the total leader length, but may introduce crossings among the leaders  (Figure~\ref{fig:split2}(c)). However, at the end of all these swaps, we only have $\mathcal{L}_w$ and a set  $\mathcal{L}  \subseteq \mathcal{L}_{w'}$ of $w$ leaders crossing $L$.

Let $L'$ be the vertical line through  $q_w$. Assume that there are  $\alpha$   leaders that are incident to the ports in $\Gamma_1$ and cross $L'$. Recall that  the 1-sided problem determined by $p'$ and the leader connecting $q_w,s_w$   has a feasible solution. Therefore, the rectangle $B$ bounded by $R(t), R(l)$, $L'$ and the blocker is balanced (e.g., see Figure~\ref{fig:split2}(d)). Thus there must be exactly $\alpha$  leaders that connect ports from the left of $L'$ and enter $B$ (e.g., the orange leaders in Figure~\ref{fig:split2}(d)). We can swap the sites of these two sets of leaders (as we did for $L$) without increasing the total leader length (e.g., see Figure~\ref{fig:split2}(e)). It is now straightforward to compute a planar solution inside $B$ using Benkert et al.'s algorithm~\cite[Lemma 1]{DBLP:journals/jgaa/BenkertHKN09}. 

Note that the newly formed $\mathcal{L}$, i.e., the leaders that are  incident the ports on the right half-plane of $L$ and cross $L$, now lie entirely in $\Gamma_2$. Therefore, we can swap the sites of the sets  $\mathcal{L}_w$ and $\mathcal{L}$, and then remove the crossings on the left and right half-planes of $L$ independently  using  Benkert et al.'s algorithm~\cite[Lemma 1]{DBLP:journals/jgaa/BenkertHKN09}. 

\textbf{Case 3:} The remaining case is when $\Gamma_3$ is empty, which is straightforward to process in the way we handled Case 1.  
\end{proof}

\section{Relating Boundary Labeling to Outerstring Graphs}
In this section, we reduce the boundary labeling problem to the independent set problem on a class of weighted geometric intersection graphs in the plane called outerstring graphs. We show that if one can discretize  a boundary labeling problem such that the number of candidate leaders is a polynomial in $n$, then our approach will yield a polynomial-time algorithm for the problem.

An \emph{outerstring graph} is an  intersection graph of a set of curves in the Euclidean plane that lie inside a polygon such that one of the endpoints of each curve is attached to the boundary of the polygon. 
 Keil et al.~\cite{KeilMPV17} gave an $O(N^3)$-time algorithm for the maximum-weighted independent set problem on outerstring graphs. The algorithm requires an outerstring graph as an input, where   each curve is given as a polygonal line (i.e., a chain of straight line segments) and $N$ is the number of segments in the representation. We show that by discretizing the boundary labeling problem and assigning an appropriate weight to each candidate leader, one can reduce the boundary labeling problem to the maximum-weighted independent set problem on outerstring graphs. As an example, here we show the reduction for the boundary labeling problem using $opo$-leaders in the presence of obstacles.

\subparagraph{Boundary Labeling with  Orthogonal Obstacles:} Fink and Suri~\cite{DBLP:conf/cccg/0001S16} gave $O(n^9)$ and $O(n^{21})$-time algorithms for the opposite 2-sided boundary labeling with $po$ and $opo$-leaders, respectively. Our approach will yield $O(n^6)$ and $O(n^{12})$-time algorithms for $po$- and $opo$-leaders, respectively, irrespective of the labeling model (opposite, adjacent, or for any port distribution on the boundary). 
 For the opposite 2-sided case,  the running time reduces to $O(n^6)$ and $O(n^{9})$ (for $po$- and $opo$-leaders, respectively). This will settle the time complexity question of 1-bend 3- and 4-sided boundary labeling~\cite{DBLP:journals/algorithmica/KindermannNRS0W16}.  In the rest of this section, we relax the general position assumption and denote $n$ to be the total number of sites and obstacle vertices.

First consider the case of $po$-leaders. Let $I$ be an instance of the boundary labeling. Given a site and a port, there is at most one way of connecting them. Let $M$ denote the set of all possible leaders that do not intersect any obstacle. Then $|M|\in O(n^2)$. It is straightforward to compute $M$ in $O(n^3)$ time. Observe that  each leader $l\in M$ can be viewed as an outerstring, and let $st(l)$ be the corresponding outerstring. Let $\len{l}$ be the length of the leader $l$, and define $x:=\max_{l\in M}\len{l}$ and $y:=\min_{l\in M}\len{l}$. Let $C$ be a number such that $C > nx-(n-1)y> 0$. For each leader $l\in M$, we assign a weight $w(st(l))$ to $st(l)$, where  $w(st(l)):=C-\len{l}$. The following lemma and Keil et al.'s~\cite{KeilMPV17} algorithm  lead  us to the results for $po$-leaders (Theorem~\ref{thm:Opposite2Sided}).
\begin{lemma}
\label{lem:outerstring}
$I$ has a feasible solution with total leader length $L$ if and only if the corresponding outerstring graph $G_I$ has a feasible solution with total weight $(nC-L)$.
\end{lemma}
\begin{proof}
A feasible solution $S$ of $I$ with total leader length $L$ gives a feasible solution for $G_I$ with total weight
$
\sum_{l\in S}w(st(l))=\sum_{l\in S}(C-\len{l})=nC-\sum_{l\in S}\len{l}=(nC-L).
$

We now assume  that $G_I$ has a feasible solution $S'$ with total weight $W = (nC-L)$, and show that the corresponding leaders $S$ yields a feasible solution of $I$ of total leader length $L$. Since $S'$ is an independent set, the  leaders in $S$ are crossing-free,  as well as no site or port is incident to more than one leader. It now suffices to show that every site is connected to a string, i.e., $|S|=n$ and the total leader length is $L$. Observe that $W = nC-L \geq nC-nx > nC-(n-1)y-C = (n-1)(C-y)$. If $|S|< n$, then the total leader length is at most $(n-1)x$, and $S'$ has weight at most $(n-1)(C-y)$, which contradicts that $W> (n-1)(C-y)$. Therefore, $|S|=n$, and we have
\[
W=\sum_{s\in S'}w(s)=\sum_{s\in S'}(C-\len{l_i})=nC-\sum_{l\in S'}\len{l}.
\] 
Since $W=(nC-L)$, we have $\sum_{l\in S'}= L$.
\end{proof}

\begin{theorem}
\label{thm:Opposite2Sided}
The boundary labeling problem can be solved in $O(n^6)$ time using $po$-leaders, for both adjacent and opposite sided models, even in the presence of obstacles (where $n$ is  the total number of sites and vertices of the obstacles).
\end{theorem}

Consider now the case for $opo$-leaders. For opposite 2-sided case, Fink and Suri~\cite{DBLP:conf/cccg/0001S16} showed that one can discretize the problem such that if there exists a feasible solution, then there is one where the $x$-coordinate of the middle segment of every leader lies in the set of all $x$-coordinates of the sites and obstacle vertices. Therefore, we have $O(n)$ potential leaders for each port-site pair, and thus $O(n^3)$ leaders in total. 
 Hence applying Keil et al.'s~\cite{KeilMPV17} algorithm gives a running time of $O(n^{9})$.
 
The discretization of~\cite{DBLP:conf/cccg/0001S16}  does not apply to the 3- and 4-sided case. However, consider a grid $H$ determined by the  axis-aligned lines  through the ports, sites and obstacle vertices. For each pair of consecutive parallel lines of $H$, place a set of $n$ parallel lines in between. Let the resulting grid be $H'$. If there is a feasible solution to the boundary labeling problem, then for any pair of consecutive parallel vertical lines $\ell,\ell'$ (similarly for horizontal) of $H$, we can have at most $n$ middle vertical segments of the leaders. We thus can distribute them by moving  horizontally to the $n$ lines of $H'$ (e.g.,  see~\cite{DBLP:conf/cccg/0001S16}), which does not change the total leader length. By construction, there is no site, port or obstacle vertex between $\ell$ and $\ell'$. Hence such a modification can be performed without introducing any crossing. Since $H'$ is an $O(n^2)\times O(n^2)$ grid and since we have $O(n^2)$ potential leaders for each port-site pair, the number of candidate leaders is  $O(n^4)$. Hence applying Keil et al.'s~\cite{KeilMPV17} algorithm gives a running time of $O(n^{12})$.

\begin{theorem}
\label{thm:2Sided}
The adjacent boundary labeling problem can be solved in $O(n^{12})$ time using $opo$-leaders, even in the presence of obstacles (where $n$ is the total number of sites and vertices of the obstacles). For opposite 2-sided models, the running time reduces to $O(n^9)$.  
\end{theorem}

\subparagraph{Sliding Port and Bend Minimization:} 
The outerstring-graph approach can also be applied to the sliding port model, where each label is assigned a distinct  interval on the boundary of $R$ and a site can be connected to any point of an interval. The goal here is to minimize the total leader length or the number of bends. We only need to discretize the problem such that the number of strings that we need to consider is a polynomial in $n$. Define $H$ to be a grid determined by the axis-aligned lines through sites, interval boundaries and obstacle vertices. Construct $H'$ from $H$ by introducing for  every pair of consecutive parallel lines of $H$, a set of $2n$ parallel  lines in between. 

The grid $H'$ can be used to discretize the problem, as follows. The segments incident to the sites are already on $H$. Consider now  a vertical (similarly for horizontal) segment $\ell$  that is incident to an interval $I$, but not incident to any site. Let $\ell'$ and $\ell'$ be a pair of consecutive horizontal lines of $H$ such that $\ell$ lies between them. There can be at most $2n$ horizontal lines  between $\ell,\ell'$, which we can distribute to the lines of $H'$ by moving vertically (e.g.,  see~\cite{DBLP:conf/cccg/0001S16}). Since there cannot be any site, interval boundary or obstacle vertex between $\ell,\ell'$, such a modification neither introduce crossings nor increase the number of total bends. By the construction of $H$, the boundary of $R$ between $\ell,\ell'$ lies in the interval $I$. Hence  $\ell$ will still be incident to $I$. Finally, the middle segments of the leaders can be processed in the same way as we did for Theorem~\ref{thm:2Sided}. It is straightforward to observe that the number of potential strings is a polynomial in $n$. We can now assign certain weights to these strings such that the maximum-weight independent set of the corresponding outerstring graph yields a minimum-bend solution for the boundary labeling problem.

We first consider the case of $po$-leaders. Let $I$ be an instance of this problem. Consider the set $M$ of outerstrings as before. For each outerstring $st(l)\in M$, we assign the weight $w(st(l))$, where
\begin{equation}
  w(st(l))=\begin{cases}
    n+2, & \text{if $l$ has no bends}.\\
    n+1, & \text{if $l$ has one bend}.
  \end{cases}
\end{equation}
This forms our instance $G_I$ of an outerstring graph on which we solve the maximum-weighted independent set problem by running  Keil et al.'s  algorithm~\cite{KeilMPV17}.
\begin{lemma}
\label{lem:poBendMinimization}
Let $I$ be an instance of the boundary labeling problem with $po$-leaders. Then $I$ has a feasible solution  with $k$ bends if and only if the instance $G_I$ has a feasible solution $W$ with total weight at least $(n^2+2n-k)$.
\end{lemma}
\begin{proof}
Let $S$ be a feasible solution of $I$. Clearly, the strings corresponding to the leader of $S'$ is a feasible solution for $G_I$. Let $k$ be the total number of bends in $S$. Then the weight of $S'$ is  
\[
\sum_{l\in S}w(l)\geq (n+1)k+(n+2)(n-k)=n^2+2n-k.
\]

Assume now that $G_I$ has a feasible solution $S$ with weight at least $n^2+2n-k$.  Let $S'$ be the corresponding set of leaders in $I$. Since $S$ is an independent set, a port or site can be incident to at most one leader of $S$. If a site is not connected to any port in $S'$, then at most $(n-1)$ sites are incident to a leader. Since the maximum weight of a leader can be at most $(n+2)$, the weight of $S$ is at most $(n-1)(n+2)=(n^2+n-2)$, which is a contradiction since the weight of $S$ is at least $n^2+2n-k > (n^2+n-2)$ (because $n\ge k$). Therefore, $|S'|=n$.

It now remains to show that the weight of $S'$ is at most $k$. Suppose for a contradiction that $S'$ has at least $(k+1)$ $po$-leaders. Therefore, the weight of $S$ is at most $(n+1)(k+1)+(n+2)(n-k-1)=n^2+2n-(2k+1)< (n^2+2n-k)$, which is a contradiction that the weight of $S$ is at least $(n^2+2n-k)$.
\end{proof}

Now, we consider the case of $opo$-leaders. Let $I$ be an instance of this problem. Consider the set $M$ of outerstrings as before. For each outerstring $st(l)\in M$, we assign the weight $w(st(l))$ as follows:
\begin{equation}
  w(st(l))=\begin{cases}
    \alpha+3, & \text{if $l$ has no bends},\\
    \alpha+2, & \text{if $l$ has one bend},\\
    \alpha+1, & \text{if $l$ has two bends}.\\
  \end{cases}
\end{equation}
Here $\alpha =2n$. 

\begin{lemma}
\label{lem:opoBendMinimization}
Let $I$ be an instance of the boundary labeling problem with $opo$-leaders. Then $I$ has a feasible solution  with $k$ bends if and only if the instance $G_I$ has a feasible solution $W$ with total weight at least $(\alpha n + 3n - k)$.
\end{lemma}
\begin{proof}
Let $S$ be a feasible solution of $I$. Clearly, the strings corresponding to the leader of $S'$ is a feasible solution for $G_I$. Let $k$ be the total number of bends in $S$, and let $k_1$ and $k_2$ be the number of strings with 1-bend and 2-bends, respectively. Therefore, $k_1+2k_2 = k$, and  the weight of $S'$ is  
\[
\sum_{l\in S}w(l)= (\alpha+2)k_1+(\alpha+1)k_2+(\alpha+3)(n-k_1-k_2)=    \alpha n+3n -k_1 -2 k_2  =   \alpha n + 3n - k.
\]

Assume now that $G_I$ has a feasible solution $S$ with weight at least $(\alpha n + 3n - k)$.  Let $S'$ be the corresponding set of leaders in $I$. Since $S$ is an independent set, a port or site can be incident to at most one leader of $S$. If a site is not connected to any port in $S'$, then at most $(n-1)$ sites are incident to a leader. Since the maximum weight of a leader can be at most $(\alpha+3)$, the weight of $S$ is at most $(n-1)(\alpha +3)=(\alpha n+3n-\alpha -3)$, which is a contradiction since the weight of $S$ is at least $\alpha n + 3n - k > (\alpha n+3n-\alpha -3)$ (because $\alpha = 2n \ge k$). Therefore, $|S'|=n$.

It now remains to show that the leaders  of $S'$ has at most $k$ bends. Suppose for a contradiction that $S'$ has at least $k_1$ $po$-leaders and $k_2$ $opo$-leaders such that $k_1+2k_2 \ge k+1$ . Therefore, the weight of $S$ is at most $(\alpha+2)k_1+(\alpha+1)k_2+ (\alpha+3)(n-k_1-k_2)=\alpha n+3n-(k_1+2k_2)-(2k_1+k_2)+(2k_1+k_2)=\alpha n+3n-(k_1+2k_2)$. Since $k_1+2k_2\geq k+1$, the weight of $S$ is strictly less than $\alpha n+3n-k$, which is a contradiction.
\end{proof}

By Lemmas~\ref{lem:poBendMinimization} and~\ref{lem:opoBendMinimization}, we have the following theorem (which settles two open questions of~\cite[Table 23.1]{handbook}).
\begin{theorem}
\label{thm:obsOpposite2Sided}
A boundary labeling that minimizes the total number of bends can be computed (if exists) in polynomial time for  both adjacent and opposite models (with sliding ports, $po$ and $opo$-leaders), even in the presence of obstacles. 
\end{theorem}

\section{Conclusion}
The most natural directions for future research is to improve the time complexity of our algorithm for the 1-bend adjacent 2-sided model. A number of intriguing questions follow: Can we find a non-trivial lower bound on the time-complexity? Is the problem 3-sum hard or, as hard as `sorting $X+Y$'? Can we check the feasibility in near-linear time? It would also be interesting to find fast approximation algorithms for boundary labeling problems. 
\bibliographystyle{plain}
\bibliography{ref}

\end{document}